\documentclass[runningheads]{llncs}
\usepackage{iftex}
\ifLuaTeX\else\pdfoutput=1\fi
\usepackage{amsmath}
\usepackage[final]{listings}
\usepackage[prologue,gray]{xcolor}
\usepackage{hyperref}
\usepackage{cleveref}
\usepackage{braket}
\usepackage{tikz}
\usepackage{mathtools}
\usepackage{amssymb}
\usepackage[frozencache]{minted}
\ifLuaTeX
\usepackage{fontspec}
\else
\usepackage{inconsolata}
\fi
\usepackage[backend=biber,style=lncs,doi=false]{biblatex}
\usepackage{stmaryrd}
\usetikzlibrary{matrix,decorations,arrows,calc,trees}
\usepackage{tcolorbox}
\tcbuselibrary{listings,xparse}
\usepackage[final]{fixme}
\usepackage{enumitem}
\newcommand{\coloneq}{\mathrel{:=}}
\makeatletter

\newtheorem{algorithm}{Algorithm}
\AtBeginDocument{%
  \providecommand\BibTeX{{%
    \normalfont B\kern-0.5em{\scshape i\kern-0.25em b}\kern-0.8em\TeX}}}

\newcommand{\R}{\mathbb{R}}
\newcommand{\N}{\mathbb{N}}
\newcommand{\rem}[2]{{\overline{#1}}^{#2}}
\newcommand{\X}{\boldsymbol{X}}
\newcommand{\Y}{\boldsymbol{Y}}
\newcommand{\Rseries}{\mathord{\R\llbracket\X\rrbracket}}
\DeclareMathOperator{\NonVan}{NV}
\DeclareMathOperator{\RF}{RF}
\newcommand{\Tower}{\mathord{\mathrm{Tower}}}
\newcommand{\smoothto}{\xrightarrow{C^\infty}}

\definecolor{ltblue}{rgb}{0,0.4,0.4}
\definecolor{dkblue}{rgb}{0,0.1,0.6}
\definecolor{dkgreen}{rgb}{0,0.35,0}
\definecolor{dkviolet}{rgb}{0.3,0,0.5}
\definecolor{dkred}{rgb}{0.5,0,0}
\definecolor{comment}{HTML}{444444}
\definecolor{keywd}{HTML}{8D00ED}
\definecolor{types}{HTML}{1F7B2F}
\definecolor{str}{HTML}{4070a0}
\definecolor{code-background}{gray}{0.8}
\definecolor{pragma}{HTML}{372A78}
\definecolor{num}{HTML}{40a070}
\definecolor{symb}{HTML}{000000}
\newcommand\symbmath[1]{{\ensuremath{\color{symb}#1}}}

\def\ul#1{{\underline{\bfseries #1}}}
\newcommand{\divs}{\mathrel{\mid}}
\newcommand{\CartSp}{\mathord{\mathrm{CartSp}}}

\newcommand{\Sets}{\mathord{\mathrm{Sets}}}

\newminted[code]{haskell}{%
mathescape,linenos,numbersep=5pt,frame=lines,framesep=2mm,%
}
\lstdefinelanguage{algorithm}{
  morecomment=[l][\slshape\color{comment}]{//},
  morecomment=[n][\bfseries\itshape]{/*}{*/},
  morestring=[b]\",
  stringstyle={\slshape},
  keywords={Pick,For,Done,End,in,with,Let,While,Loop,Do,Then,Else,Until,Unless,Return,INPUT,OUTPUT,If},
  numbers=left,
  escapechar=\@,
  keywordstyle={\ul},
  morekeywords=[2]{*,+,/,-,.,\%,!,?,=,<,>,$},
  keywordstyle=[2]{\ttfamily\bfseries},
  morekeywords=[3]{Convol,WeilTest,LiftWeil,LiftSeriesPoly,TowerAD},
  keywordstyle=[3]{\rmfamily\upshape\scshape},
  emphstyle={[2]\itshape},
  literate=
    {/=}{{\symbmath{\neq}}}2
    {forall}{{\ensuremath{\pmb{\forall}}}}1
    {*}{{\symbmath{\times}}}1
    {I}{{\symbmath{I}}}1
    {+}{{\symbmath{+}}}1
    {>}{{\symbmath{>}}}2
    {...}{{\symbmath{\ldots}}}2
    {<}{{\symbmath{<}}}2
    {alpha}{{\symbmath{\alpha}}}1
    {beta}{{\symbmath{\beta}}}1
    {->}{{\symbmath{\rightarrow}}}2 {>=}{{\symbmath{\geq}}}2 {<-}{{\symbmath{\leftarrow}}}2
    {g_i}{{\symbmath{g_i}}}2
    {f_i}{{\symbmath{f_i}}}2
    {===}{{\symbmath{\equiv}}}2
    {==>}{{\symbmath{\Longrightarrow}}}{3}
    {\\}{{\symbmath{\lambda}}}1
    {<=}{{\symbmath{\leq}}}1 {=>}{{\symbmath{\Rightarrow}}}2
}
\lstnewenvironment{alg}[1][frame=leftline]{\lstset{language=algorithm}}{}

\VerbatimFootnotes
\DeclareTotalTCBox{\hask}{v}{%
tcbox raise base,box align=base,verbatim,colback=lightgray,colframe=gray%
}{\mintinline[fontsize=\tiny]{haskell}{#1}}

\newminted[repl]{haskell}{%
mathescape,numbersep=5pt,frame=lines,framesep=2mm,%
}

\DeclareTotalTCBox{\hask}{v}{box align=base,verbatim,colback=lightgray,colframe=gray}{\lstinline[language=haskell]{#1}}

\begin{document}
\title{Automatic Differentiation With Higher Infinitesimals, or Computational Smooth Infinitesimal Analysis in Weil Algebra\thanks{This work was supported by the Research Institute for Mathematical Sciences, an International Joint Usage/Research Center located in Kyoto University.
}}
\titlerunning{Automatic Differentiation With Higher Infinitesimals}
%
\author{Hiromi Ishii\inst{1}}
\authorrunning{H. Ishii}
%
\institute{DeepFlow, Inc., 3-16-40 Fujimi-shi Tsuruse nishi 354-0026, Japan}
\maketitle              
\begin{abstract}
  \frenchspacing
  We propose an algorithm to compute the $C^\infty$-ring structure of arbitrary Weil algebra.
  It allows us to do some analysis with \emph{higher infinitesimals} numerically and symbolically.
  To that end, we first give a brief description of the (Forward-mode) \emph{automatic differentiation} (AD) in terms of \emph{$C^\infty$-rings}.
  The notion of a $C^\infty$-ring was introduced by Lawvere~\cite{lawvere1979categorical} and used as the fundamental building block of \emph{smooth infinitesimal analysis} and \emph{synthetic differential geometry}~\cite{Moerdijk:1991aa}.
  We argue that interpreting AD in terms of $C^\infty$-rings gives us a unifying theoretical framework and modular ways to express multivariate partial derivatives.
  In particular, we can ``package'' higher-order Forward-mode AD as a Weil algebra, and take tensor products to compose them to achieve multivariate higher-order AD.
  The algorithms in the present paper can also be used for a pedagogical purpose in learning and studying smooth infinitesimal analysis as well.
\end{abstract}


\section{Introduction}\label{sec:intro}
\sloppy
\emph{Automatic Differentiation} (or, \emph{AD} for short) is a method to calculate derivatives of (piecewise) smooth functions accurately and efficiently.
AD has a long history of research, and under the recent rise of differentiable programming in machine learning, AD has been attracting more interests than before recently.

\emph{Smooth Infinitesimal Analysis} (or, \emph{SIA} for short), on the other hand, is an area of mathematics that uses \emph{nilpotent infinitesimals} to develop the theory of real analysis.
Its central building blocks are \emph{Weil algebras}, which can be viewed as the real line augmented with nilpotent infinitesimals.
Indeed, SIA is a subarea of \emph{Synthetic Differential Geometry} (SDG) initiated by Lawvere~\cite{lawvere1979categorical}, which studies smooth manifolds topos-theoretically, and higher multivariate infinitesimals play crucial roles in building theory of, e.g.\ vector fields, differential forms and tangent spaces.
The key observation of Lawvere is that manifolds can be classified solely by their smooth function ring $C^\infty(M)$, and both such function rings and Weil algebras are special cases of \emph{$C^\infty$-rings}.

It has been pointed out that AD and SIA have some connection; e.g.\ even Wikipedia article~\cite{Wikipedia:2021aa} mentions the connection between first-order Forward-mode AD with the ring $\R[X]/X^2$ of dual numbers.
However, a precise theoretical description of this correspondence is not well-communicated, and further generalisation of AD in terms of SIA hasn't been discussed in depth.

The present paper aims at filling this gap, giving a unified description of AD in terms of $C^\infty$-rings and Weil algebras.
Furthermore, our main contribution is algorithms to compute the $C^\infty$-ring structure of a general Weil algebra.
This enables automatic differentiation done in \emph{arbitrary} Weil algebras other than dual numbers, and, together with tensor products, lets us compute higher-order multivariate partial derivatives in a modular and composable manner, packed as Weil algebra.
Such algorithms can also be used to learn and study the theory of SIA and SDG.

This paper is organised as follows.
In \Cref{sec:prel}, we review the basic concepts and facts on $C^\infty$-rings and Weil algebras.
This section provides basic theoretical background --- but the proofs of proposed algorithms are, however, not directly dependent on the content of this section.
So readers can skip this section first and go back afterwards when necessary.
Subsequently, we discuss the connection between Forward-mode automatic differentiation and Weil algebras in \Cref{sec:ad-and-weils}.
There, we see how the notion of Weil algebra and $C^\infty$-ring can be applied to treat higher-order partial ADs in a unified and general setting.
Then, in \Cref{sec:alg}, we give algorithms to compute the $C^\infty$-ring structure of an arbitrary Weil algebra. These algorithms enable us to do \emph{automatic differentiation with higher infinitesimals}, or \emph{computational smooth infinitesimal analysis}.
We give some small examples in \Cref{sec:examples}, using our proof-of-concept implementation~\cite{Ishii:2020aa} in Haskell.
Finally, we discuss related and possible future works and conclude in \Cref{sec:concl}.

\section{Preliminaries}\label{sec:prel}
In this section, we briefly review classical definitions and facts on Weil algebras and $C^\infty$-rings without proofs, which will be used in \Cref{sec:alg}.
For theoretical detail, we refer readers to Moerdijk--Reyes~\cite[Chapters I and II]{Moerdijk:1991aa} or Joyce~\cite{joyce2016algebraic}.

We use the following notational convention:
\begin{definition}[Notation]
  Throughout the paper, we use the following notation:
  \begin{itemize}
    \item
    $g \circ f$ denotes the composite function from $A$ to $C$ of functions $f: A \to B$ and $g: B \to C$, that is, the function defined by $(g \circ f)(x) = g(f(x))$ for all $x \in A$.
    \item For functions $f_i: Z \to X_i\,(1 \leq i \leq n)$,
    $\braket{f_1, \dots, f_n}$ denotes the product of functions $f_i$ given by the universality of the product objects.
    That is, $\braket{f_1, \dots, f_n}$ is the function of type $Z \to X_1 \times \cdots \times X_n$ defined by $\braket{f_1, \dots, f_n}(z) = (f_1(z), \dots, f_n(z)) \in X_1 \times \cdots \times X_n$ for all $z \in Z$
  \end{itemize}
\end{definition}

\begin{definition}[Lawvere~\cite{lawvere1979categorical}]
  A \emph{$C^\infty$}-ring $A$ is a product-preserving functor from the category $\CartSp$ of finite-dimensional Euclidean spaces and smooth maps to the category $\Sets$ of sets.

  We identify $A$ with $A(\R)$ and $A^n$ with $A(\R^n)$.
  For a map $f: \R^m \to \R$, we call $A(f): A^m \to A$ the \emph{$C^\infty$-lifting} of $f$ to $A$.
\end{definition}

Intuitively, a $C^\infty$-ring $A$ is an $\R$-algebra $A$ augmented with $m$-ary operations $A(f): A^m \to A$ respecting composition, projection and product for all smooth maps $f: \R^m \to \R$.

One typical example of a $C^\infty$-ring is a formal power series ring:

\begin{theorem}[{Implicitly in Lawvere~\cite{lawvere1979categorical}; See~\cite[1.3  Borel's Theorem]{Moerdijk:1991aa}}]\label{thm:series-is-smooth}
  The ring $\R\llbracket X_1, \dots, X_n\rrbracket$ of formal power series with finitely many variables has the $C^\infty$-ring structure via Taylor expansion at $\boldsymbol{0}$.
  In particular, lifting of a smooth map $f: \R^m \to \R$ is given by:
  \[
    \Rseries(f)(g_1, \dots, g_m) = \sum_{\alpha \in \N^n} \frac{\X^\alpha}{\alpha!} D^\alpha(f \circ \braket{g_1, \dots, g_m})(\boldsymbol{0}),
  \]
  where $\alpha! = \alpha_1 ! \dots \alpha_n !$ is the multi-index factorial and $D^\alpha$ is the partial differential operator to degree $\alpha$.
\end{theorem}

The $C^\infty$-rings of central interest in this paper are \emph{Weil algebras}, and have a deep connection with $\Rseries$:

\begin{definition}[Weil algebra]
  A \emph{Weil algebra} $W$ is an associative $\R$-algebra which can be written as $W = \R[X_1, \dots, X_n]/I$ for some ideal $I \subseteq \R[\X]$ such that $\braket{X_1, \dots, X_n}^k \subseteq I$ for some $k \in \N$.
\end{definition}
It follows that a Weil algebra $W$ is finite-dimensional as a $\R$-linear space and hence $I$ is a \emph{zero-dimensional} ideal.
A Weil algebra $W$ can be regarded as a real line $\R$ augmented with nilpotent infinitesimals $d_i = {[X_i]}_I$.
In what follows, we identify an element $\boldsymbol{u} \in W$ of a $k$-dimensional Weil algebra $W$ with a $k$-dimensional vector $\boldsymbol{u} = (u_1, \dots, u_k) \in \R^k$ of reals.

Although it is unclear from the definition, Weil algebras have the canonical $C^\infty$-structure.
First note that, if $I$ is zero-dimensional, we have $\R[\X]/I \simeq \R\llbracket \X \rrbracket /I$.
Hence, in particular, any Weil algebra $W$ can also be regarded as a quotient ring of the formal power series by zero-dimensional ideal.
Thus, together with \Cref{thm:series-is-smooth}, the following lemma shows that any Weil algebra $W$ has the canonical $C^\infty$-ring structure:

\begin{lemma}[{Implicitly in Lawvere~\cite{lawvere1979categorical}; See~\cite[1.2 Proposition]{Moerdijk:1991aa}}]\label{lem:quot-ring-ideal}
  For any $C^\infty$-ring $A$ and a ring-theoretical ideal $I \subseteq A$, the quotient ring $A/I$ again has the canonical $C^\infty$-ring structure induced by the canonical quotient mapping:
  \[
    (A/I)(f)([x_1]_I, \dots, [x_m]_I) \coloneqq \left[ A(f)(x_1, \dots, x_m) \right]_I,
  \]
  where $x_i \in A$ and $f: \R^m \smoothto \R$.
  In particular, the $C^\infty$-structure of Weil algebra $W$ is induced by the canonical quotient mapping to that of $\Rseries$.
\end{lemma}

\section{Connection between Automatic Differentiation and Weil Algebras}
\label{sec:ad-and-weils}
In this section, based on the basic facts on $C^\infty$-rings and Weil algebras reviewed in \Cref{sec:prel}, we describe the connection of automatic differentiation (AD) and Weil algebra.

\emph{Forward-mode} AD is a technique to compute a value and differential coefficient of given univariate composition of smooth function efficiently.
It can be implemented by ad-hoc polymorphism (or equivalently, function overloading).
For detailed implementation, we refer readers to Elliott~\cite{Elliott2009-beautiful-differentiation} and Kmett's \texttt{ad} package~\cite{Kmett:2010aa}.

Briefly speaking, in Forward-mode AD, one stores both the value and differential coefficient simultaneously, say in a form $f(x) + f'(x) d$ for $d$ an indeterminate variable.
Then, when evaluating composite functions, one uses the Chain Rule for implementation:
\[
  \frac{\mathrm{d}}{\mathrm{d}x}(g \circ f)(x) = f'(x) g'(f(x)).
\]

The following definitions of functions on dual numbers illustrate the idea:

\begin{align*}
  (a_1 + b_1 d) + (a_2 + b_2 d) &= (a_1 + a_2) + (b_1 + b_2)d\\
  (a_1 + b_1 d) \times (a_2 + b_2 d) &= a_1 a_2 + (a_1 b_2 + a_2 b_1)d\\
  \cos(a_1 + b_1 d) &= \cos(a_1) - b_1 \sin(a_1) d
\end{align*}

The last equation for $\cos$ expresses the nontrivial part of Forward-mode AD.
As mentioned above, we regard $a_1 + b_1 d$ as a pair $(a_1, b_1) = (f(x), f'x)$ of value and differential coefficient of some smooth function $f$ at some point $x$.
So if $a_2 + b_2 d = \cos(a_1 + b_1 d)$, we must have $a_2 = \cos(f(x)) = \cos a_1$ and $b_2 = \frac{\mathrm d}{\mathrm{d} x} \cos(f(x)) = -b_1 \sin(a_1)$ by Chain Rule.
The first two equations for addition and multiplication suggest us to regard operations on Forward-mode AD as extending the algebraic structure of $\R[d] = \R[X]/X^2$.
Indeed, first-order Forward-mode AD can be identified with the arithmetic on \emph{dual numbers}:
\begin{definition}
  The \emph{dual number ring} is a Weil algebra $\R[X]/{X^2}$.
  We often write $d = {[X]}_I \in \R[d]$ and $\R[d] \coloneq \R[X]/{X^2}$.

  We use an analogous notation for multivariate versions:
  
  \[
    \R[d_1, \dots, d_k] \coloneq \R[\X]/\braket{X_1^2, \dots, X_k^2}.
  \]
\end{definition}

Since the dual number ring $\R[d]$ is a Weil algebra, one can apply \Cref{thm:series-is-smooth} and \Cref{lem:quot-ring-ideal} to compute its $C^\infty$-structure.
Letting $f: \R \smoothto \R$ be a univariate smooth function, then we can derive the $C^\infty$-lifting $\R[d](f): \R[d] \to \R[d]$ as follows:

\begin{alignat*}{3}
  &&&\R\llbracket{}X\rrbracket(f)(a + bX) \\
  &&=\:&f(a) + \frac{\mathrm d}{\mathrm{d} x}(f(a + bx))(0) X + \cdots
  &\quad&(\text{by \Cref{thm:series-is-smooth}})\\
  &&=\:&f(a) + b f'(a) X + \cdots\\
  &&\xrightarrow{X \mapsto d}\:&f(a) + bf'(a) d,\\
  &\therefore&& \R[d](f)(a + bd) = f(a) + b f'(a) d.
  && (\text{by \Cref{lem:quot-ring-ideal}}) \tag{\ensuremath{\ast}}
  \label{eqn:dual-smooth}
\end{alignat*}

One can notice that the derived $C^\infty$-structure in \eqref{eqn:dual-smooth} is exactly the same as how to implement individual smooth functions for Forward-mode AD.
This describes the connection between Forward-mode AD and dual numbers: Forward-mode AD is just a (partial) implementation of the $C^\infty$-structure of the dual number ring $\R[d]$.

Let us see how this extends to higher-order cases.
The most na\"{i}ve way to compute higher-order derivatives of smooth function is just to successively differentiating it.
This intuition can be expressed by duplicating the number of the basis of dual numbers:
\begin{theorem}\label{thm:univ-partial-duals}
  For any $f: \R \smoothto \R$ and $\boldsymbol{x} \in \R^n$, we have:
  \[
    \R[d_1, \dots, d_k](f)(x + d_1 + \cdots + d_n) 
    = \sum_{0 \leq i \leq n} f^{(i)}(x)\sigma^i_n(\vec{d}),
  \]
  where, $\sigma^i_k(x_1, \dots, x_k)$ denotes the $k$-variate elementary symmetric polynomial of degree $i$.
\end{theorem}

The above can be proven by an easy induction.

However, as one can easily see, terms in $\Rseries/\braket{X_i^2}_i$ can grow exponentially and include duplicated coefficients.
How could we reduce such duplication and save space consumption? ---this is where general Weil algebras beyond (multivariate) dual numbers can play a role.
We can get derivatives in more succinct representation with \emph{higher infinitesimal} beyond dual numbers:

\begin{lemma}\label{lem:higher-infinitesimal}
  Let $I = \braket{X^{n + 1}}, W = \R[X]/I$ and $\varepsilon = {[X]}_I$ for $n \in \N$.
  Given $f: \R \smoothto \R$ and $a \in \R$, we have:
  \[
    W(f)(a + \varepsilon)
    = \sum_{k \leq n} \frac{f^{(k)}(a)}{k !} \varepsilon^k.
  \]
\end{lemma}
In this representation, we have only $(n + 1)$-terms, and hence it results in succinct and efficient representation of derivatives.

If we duplicate such higher-order infinitesimals as much as needed, one can likewise compute \emph{multivariate} higher-order derivatives all at once, up to some multidegree $\beta$:
\begin{lemma}
  Let $I = \Braket{X_i^{\beta_i + 1}\ |\ i \leq m}$, $W = \R[X_1, \dots, X_m]/I$, and $\varepsilon_i = {[X_i]}_I$ for $\beta = {(\beta_i)}_{i \leq m} \in \N^m$.
  For $f: \R^m \smoothto \R$ and $\boldsymbol{a} = {(a_i)}_{i \leq m} \in \R^m$, we have:
  \[
    W(f)(a_1 + \varepsilon_1, \dots, a_m + \varepsilon_m) =
      \sum_{\delta_i \leq \beta_i} 
      \frac{D^\delta f}{\delta !}(\boldsymbol{a})\ \varepsilon_1^{\delta_1} \cdots \varepsilon_m^{\delta_m}.
  \]
\end{lemma}

Note that the formal power series ring $\Rseries$ can be viewed as the inverse limit of $\R[\X]/\braket{\X^\beta}$'s.
In other words, if we take a limit $\beta_i \to \infty$, we can compute any higher derivative up to any finite orders; this is exactly what \emph{Tower-mode} AD aims at, modulo factor $\frac{1}{\beta!}$.

In this way, we can view AD as a technique to compute higher derivatives simultaneously by partially implementing a certain $C^\infty$-ring\footnote{Such implementation is inherently a partial approximation: there are $2^{\aleph_0}$-many smooth functions, but there are only countably many computable (floating) functions.}.
Forward-mode AD (of first-order) computes the $C^\infty$-structure of the dual number ring $\R[d]$; Tower-mode AD computes that of the formal power series ring $\Rseries$ (modulo reciprocal factorial).

So far, we have used a Weil algebra of form $\Rseries/I$.
So, do we need to define new ideals by hand whenever one wants to treat multiple variables?
The answer is no:

\begin{lemma}[{See~\cite[4.19 Corollary]{Moerdijk:1991aa}}]
  \label{thm:quot-tensor}
  For ideals $I \subseteq \R\llbracket\X\rrbracket$ and $J \subseteq \R\llbracket{}\Y\rrbracket$, we have:
  \[
    \Rseries/I \otimes_\R \R\llbracket{}\boldsymbol{Y}\rrbracket/J \simeq
    \R\llbracket{}\X, \boldsymbol{Y}\rrbracket/(I, J),
  \]
  where $\otimes_\R$ is a tensor product of $C^\infty$-rings.
\end{lemma}
Thanks to this lemma, we don't have to define $I$ by hand every time, but can take tensor products to compose predefined Weil algebras to compute multivariate and higher-order derivatives.
Examples of such calculations will be presented in \Cref{sec:examples}.

\section{Algorithms}\label{sec:alg}
In this section, we will present the main results of this paper: concrete algorithms to compute the $C^\infty$-structure of arbitrary Weil algebra and their tensor products.
For examples of applications of the algorithm presented here, the reader can skip to the next \cref{sec:examples} to comprehend the actual use case.

\subsection{Computing $C^\infty$-structure of Weil algebra}
\label{sec:general-weil-algs}
Let us start with algorithms to compute the $C^\infty$-structure of a general Weil algebra.
Roughly speaking, the algorithm is threefold:

\begin{enumerate}
  \item A procedure deciding Weil-ness of an ideal and returning data required to compute the $C^\infty$-structure (\textsc{WeilTest}, \Cref{alg:weil-test}),
  \item A procedure to compute the lifting  $W(f): W^m \to W$ to a Weil algebra $W$ from $\Rseries(f)$ (\textsc{LiftWeil}, \Cref{alg:smooth-weil}), and
  \item A procedure to lift smooth map $f: \R^m \to \R$ to the $n$-variate formal power series ring $\Rseries$ (\textsc{LiftSeries}, \Cref{lift-series}).\label{step:lift-series}
\end{enumerate}

We start with Weil-ness testing.
First, we define the basic data needed to compute the $C^\infty$-structure of Weil algebras:

\begin{definition}[Weil settings]
  The \emph{Weil setting} of a Weil algebra $W$ consists of the following data:
  \begin{enumerate}[ref=(\arabic*)]
    \item Monomial basis $\set{\boldsymbol{b}_1, \dots, \boldsymbol{b}_\ell}$ of $W$,
    \item $M$, the multiplication table of $W$ in terms of the basis,
    \item $(k_1, \dots, k_n) \in \mathbb{N}^n$ such that $k_i$ is the maximum satisfying $X_i^{k_i} \notin I$ for each $i$, and
    \item $\NonVan_W$, a table of representations of non-vanishing monomials in $W$;
    i.e.\ for any $\alpha = (\alpha_1, \dots, \alpha_n) \in \N^n$, if $\alpha_i \leq k_i$ for all $i$, then $\NonVan_W(\X^\alpha) = (c_1, \dots, c_n) \in \R^k$ satisfies $[\X^\alpha]_I = \sum_i c_i \boldsymbol{b}_i$.\label{item:nonvan}
  \end{enumerate}
\end{definition}

A basis and multiplication table allow us to calculate the ordinary $\R$-algebra structure of Weil algebra $W$.
The latter two data, $\vec{k}$ and $\NonVan_W$, are essential in computing the $C^\infty$-structure of $W$.
In theory, \ref{item:nonvan} is unnecessary if one stores a Gr\"{o}bner basis of $I$;
but since normal form calculation modulo $G$ can take much time in some case, we don't store $G$ itself and use the precalculated data $\NonVan$.
It might be desirable to calculate $\NonVan_W$ as lazily as possible.
Since it involves Gr\"{o}bner basis computation it is more desirable to delay it as much as possible and do in an on-demand manner.

With this definition, one can decide Weilness and compute their settings:

\begin{algorithm}[\textsc{WeilTest}]\label{alg:weil-test}
  \hspace{1em}\vspace{-.25em}
  \begin{description}
    \item[Input] An ideal $I \subseteq \mathbb{R}[X_1, \dots, X_n]$
    \item[Output] Returns the Weil settings of $W = \mathbb{R}[\boldsymbol{X}]/I$ if it is a Weil algebra; otherwise \verb|No|.
    \item[Procedure] \textup{\textsc{WeilTest}}
  \end{description}

  \begin{alg}
G <- calcGroebnerBasis(I)
If @$I$@ is not zero-dimensional
  Return No
@$\set{\boldsymbol{b}_1, \dots, \boldsymbol{b}_\ell}$@ <- Monomial basis of @$I$@
@$M$@ <- the Multiplication table of @$W$@
For i in 1..n@\label{line:weil-test:radical-start}@
  @$p_i$@ <- the monic generator of @$I \cap \R[X_i]$@
  If @$p_i$@ is not a monomial
    Return No
  @$k_i$@ <- @$\deg(p_i) - 1$@@\label{line:weil-test:radical-end}@
@$\NonVan_W$@ <- {}
For @$\alpha$@ in @$\Set{\alpha \in \N^n | \alpha_i \leq k_i \; \forall i \leq \ell}$@
  @$c_1 \boldsymbol{b}_1 + \cdots + c_\ell \boldsymbol{b}_\ell$@ <- @$\rem{\X^\alpha}{G}$@
  @$\NonVan_W(\X^\alpha)$@ <- (@$c_1, \dots, c_\ell$@)
Return (@$\vec{\boldsymbol{b}}, M, \vec{k}, \NonVan_W$@)
\end{alg}
\end{algorithm}

\begin{theorem}
  \Cref{alg:weil-test} terminates and returns expected values.
\end{theorem}
\begin{proof}
  Algorithms to decide the zero-dimensionality and calculate their multiplication table is well-known (for details, we refer readers to Cox--Little--O'Shea~\cite[Chapter 2]{CLO:2005}).
  So the only non-trivial part is nilpotence detection (\Crefrange{line:weil-test:radical-start}{line:weil-test:radical-end}).
  But, again, this is just a variation of radical calculation algorithm for zero-dimensional ideals.
  Indeed, since each $\R[X_i]$ is a PID, we have $X_i^k \in I \cap R[X_i]$ iff $p_i \divs X_i^k$, hence $p_i$ is a monomial iff $X_i$ is nilpotent in $W$.
\end{proof}
Now that we have the basis and multiplication table at hand, we can calculate the ordinary algebraic operations just by the standard means.

With upper bounds $\vec k$ of powers and representations $\NonVan_W$ of non-vanishing monomials, we can now compute the $C^\infty$-structure of an arbitrary Weil algebra, when given a lifting of smooth mapping $f$ to $\Rseries$:

\begin{algorithm}[\textsc{LiftWeil}]\label{alg:smooth-weil}
  \hfill\vspace{-.25em}
  \begin{description}
    \item[Input]
      $I \subseteq \R[\X]$, an ideal where $W = \R[\X]/I$ is a Weil algebra,
      $\R \llbracket\X\rrbracket(f): \Rseries^m \to \Rseries$, a lifting of a smooth map $f: \R^m \to \R$ to $\Rseries$, and $\vec{\boldsymbol{u}} = (\boldsymbol{u}_1, \dots, \boldsymbol{u}_m) \in W^m$,.
    \item[Output] $\boldsymbol{v} = W(f)(\vec{\boldsymbol u}) \in W$, the value of $f$ at $\vec{\boldsymbol{u}}$ given by $C^\infty$-structure.
    \item[Procedure] \textup{\textsc{LiftWeil}}
  \end{description}
\begin{alg}
(@$\vec{\boldsymbol{b}}$@, M, @$\vec{k}$@, @$\NonVan_W$@) <- WeilTest(I)
g_i <- @$(\boldsymbol{b}_1, \dots, \boldsymbol{b}_k) \cdot \boldsymbol{u}_i \in \R[\X]$@ for i <= m
h = @$\sum_\alpha c_\alpha \X^\alpha$@ <- @$\Rseries(f)(\vec{g})$@
@$\boldsymbol v$@ <- 0
For alpha with @$\alpha_i \leq k_i\, \forall i$@
  @$\boldsymbol{v}$@ <- @$\boldsymbol v$@ + @$c_\alpha \NonVan_W(\X^\alpha)$@
Return @$\boldsymbol{v}$@
\end{alg}
\end{algorithm}

The termination and validity of \Cref{alg:smooth-weil} are clear.
One might feel it problematic that \Cref{alg:smooth-weil} requires \emph{functions} as its input.
This can be \emph{any} smooth computable functions on the coefficient type.
Practically, we expect a composite function of standard smooth floating-point functions as its argument,
for example, it can be $x \mapsto \sin(x)$, $(x, y) \mapsto e^{\sin x}y^2$, and so on.
In the modern programming language -- like Haskell, Rust, LISP, Ruby, etc.\ -- one don't need to worry about their representation, as we can already freely write \emph{higher-order functions} that take functions or closures as their arguments.
Even in the low-level languages such as C/C++, one can use function pointers or whatever to pass an arbitrary function to another function.

Now that we can compute the $\R$-algebraic and $C^\infty$-structure of a Weil algebra solely from its Weil setting, one can hard-code pre-calculated Weil settings for known typical Weil algebras, such as the dual number ring or higher infinitesimal rings of the form $\R[X]/(X^{n+1})$, to reduce computational overheads.

\subsubsection{Computing the $C^\infty$-structure of $\Rseries$}\label{sec:power-series-lifting}
So it remains to compute the $C^\infty$-structure of $\Rseries$.
Thanks to \Cref{thm:series-is-smooth}, we know the precise definition of $C^\infty$-lifting to $\Rseries$:
\[
  \Rseries(f)(g_1, \dots, g_m) = \sum_{\alpha \in \N^n} \frac{\X^\alpha}{\alpha!} D^\alpha(f \circ \braket{g_1, \dots, g_n})(\boldsymbol{0}).
\]
As noted in \Cref{sec:ad-and-weils}, as a $C^\infty$-ring, the formal power series ring is isomorphic to multivariate Tower-mode AD.
It can be implemented in various ways, such as Lazy Multivariate Tower AD~\cite{Pearlmutter:2007aa}, or nested Sparse Tower AD~\cite[{module \texttt{Numeric.AD.Rank1.Sparse}}]{Kmett:2010aa}.
For reference, we include a succinct and efficient variant mixing these two techniques in \Cref{sec:appendix}.

Both Tower-mode AD and formal power series can be represented as a formal power series.
The difference is the interpretation of coefficients in a given series.
On one hand, a coefficient of $\X^\alpha$ in Tower AD is interpreted as the $\alpha$\textsuperscript{th} partial differential coefficient $D^\alpha f(\boldsymbol{a})$, where $\boldsymbol{a} = (g_1(0), \dots, g_m(0))$.
On the other hand, in $\Rseries$ it is interpreted as $\frac{D^\alpha f(\boldsymbol{a})}{\alpha!}$.
To avoid the confusion, we adopt the following convention: Tower-mode AD is represented as a function from monomials $\X^\alpha$ to coefficient $\R$ in what follows, whilst $\Rseries$ as-is.
Note that this is purely for notational and descriptional distinctions, and does not indicate any essential difference.

With this distinction, we use the following notation and transformation:

\begin{definition}
  $\Tower = \Set{f | f: \N^n \to \R }$ denotes the set of all elements of Tower-mode AD algebra.
  We denote $C^\infty$-lifting of $f: \R^m \to \R$ to $\Tower$ by $\Tower(f): \Tower^m \to \Tower$.

  A \emph{reciprocal factorial transformation} $\RF: \Tower \to \Rseries$ is defined as follows:
  \[
    \RF\left(f\right)
    = \sum_{\alpha \in \N^n} \frac{f(\alpha)}{\alpha!} \X^\alpha.
  \]
  Then, the inverse reciprocal factorial transformation is given by:
  \[
  \RF^{-1}\left(\sum_{\alpha \in \N^n} c_\alpha \X^\alpha\right)
  = \lambda (\X^\alpha).\ \alpha ! \cdot c_\alpha.
  \]
\end{definition}

\begin{algorithm}[\textsc{LiftSeries}]\label{lift-series}
\hfill\vspace{-.25em}
\begin{description}
  \item[Input]
    $f: \R^m \smoothto \R$, a smooth function which admits Tower AD,
    $g_1, \dots, g_n \in \Rseries$, formal power series.
  \item[Output] $\Rseries(f)(g_1, \dots, g_m) \in \Rseries$, $C^\infty$-lifting to the formal power series ring.
  \item[Procedure] \textup{\textsc{LiftSeries}}
\end{description}
\begin{alg}
@$\hat{g}_i$@ <- @$\RF^{-1}(g_i)$@
@$\hat{f}$@ <- @$\Tower(f)(\hat{g}_1, \dots, \hat{g}_m)$@
Return @$\RF(\hat{f})$@
\end{alg}
\end{algorithm}

\subsection{Tensor product of Weil algebras}
\label{sec:tensor-algorithm}
As indicated by \Cref{thm:quot-tensor}, tensor products enable us to compose multiple Weil algebras into one and use them to compute higher-order multivariate derivatives.
Here, we give a simple procedure to compute Weil settings of the tensor product.

\begin{algorithm}[WeilTensor]\label{alg:weil-tensor}
\hfill\vspace{-.25em}
\begin{description}
  \item[Input] Weil settings of two Weil algebras $W_1, W_2$,
  with $\set{\boldsymbol{b}_1^i, \dots, \boldsymbol{b}_{\ell_i}^i}$ a basis,
  $(k^i_1, \dots, k^i_{n_i})$ an upper bounds and $M_i$ a multiplication table for each $W_i$.
  \item[Output] Weil settings of $W_1 \otimes_\R W_2$.
  \item[Procedure] {\upshape \textsc{WeilTensor}}
\end{description}
\begin{alg}
(@$\boldsymbol{b}_1, \dots, \boldsymbol{b}_{\ell_1 \ell_2}$@) <- Convol(@$\vec{\boldsymbol{b}}^1, \vec{\boldsymbol{b}}^2$@)
M <- {}
For ({@$\boldsymbol{b}^1_L, \boldsymbol{b}^1_R$@}, @$(c_1,\dots,c_{\ell_1})$@) in @$M_1$@
  For ({@$\boldsymbol{b}^2_L, \boldsymbol{b}^2_R$@}, @$(d_1,\dots,d_{\ell_1})$@) in @$M_2$@
    M({@$\boldsymbol{b}^1_L \boldsymbol{b}^2_L$@, @$\boldsymbol{b}^1_R \boldsymbol{b}^2_R$@}) <- Convol(@$\vec{c}$@, @$\vec{d}$@)
@$\NonVan_{W_1 \otimes W_2}$@ <- {}
For (@$\X^\alpha$@, @$(c_1, \dots, c_{\ell_1})$@) in @$\NonVan_{W_1}$@
  For (@$\Y^\beta$@, @$(d_1,\dots,d_{\ell_2})$@) in @$\NonVan_{W_2}$@
    @$\NonVan_{W_1 \otimes W_2}(\X^\alpha\Y^\beta)$@ <- Convol(@$\vec{c}$@, @$\vec{d}$@)
Return @$(\boldsymbol{b}, M, (\vec{k}^1, \vec{k}^2), \NonVan_{W_1 \otimes W_2})$@
\end{alg}

Here, {\upshape \textsc{Convol}} is a convolution of two sequences:
\begin{description}
  \item[Procedure] $\mathop{\text{\upshape\scshape Convol}}((c_1, \dots, c_{\ell_1}), (d_1, \dots, d_{\ell_2}))$
\end{description}
\begin{alg}
For i in 1..(@$\ell_1 \times \ell_2$@)
  j <- @$\lfloor \frac{i}{\ell_2} \rfloor$@; k <- @$i \mod{\ell_2}$@
  @$a_i$@ <- @$c_j d_k$@
Return (@$a_1, \dots, a_{\ell_1 \ell_2}$@)
\end{alg}
\end{algorithm}

The validity proof is routine work.

\section{Examples}\label{sec:examples}
We have implemented the algorithms introduced in the previous section on top of two libraries: \texttt{computational-algebra} package~\cite{ISHII:2018ek,computational-algebra} and \texttt{ad} package~\cite{Kmett:2010aa}.
The code is available on GitHub~\cite{Ishii:2020aa}.

\subsection{Higher-order derivatives via dual numbers and higher infinitesimals}
As indicated by \Cref{thm:univ-partial-duals} and \Cref{lem:higher-infinitesimal},
to compute higher-order derivatives of univariate functions, we can use tensor products of Dual numbers or higher-order infinitesimals.

Let us first compute higher-order derivatives of $\sin(x)$ up to $3$.
First, Let us use a tensor product of dual numbers:
\begin{code}
d1, d2, d3 :: Floating a => Weil (D1 |*| D1 |*| D1) a
[d1, d2, d3] = map di [0..]
\end{code}
Here, \hask{Weil w a} represents the type of Weil algebra with its setting given in \hask{w}, \hask{D1} the dual number ideal $I = (X^2)$, and \hask{|*|} the tensor product operator.
Each $d_i$ corresponds to $i$\textsuperscript{th} infinitesimal.

Next, we calculate higher-order differential coefficients at $x = \frac{\pi}{6}$ up to the third order:

\begin{repl}
>>> (sin (pi/6),  cos (pi/6), -sin (pi/6), -cos (pi/6))
( 0.49999999999999994, 0.8660254037844387, -0.49999999999999994,
  -0.8660254037844387)

>>> sin (pi/6 + d0 + d1 + d2)
-0.8660254037844387 d(0) d(1) d(2) - 0.49999999999999994 d(0) d(1) 
  - 0.49999999999999994 d(0) d(2)  - 0.49999999999999994 d(1) d(2) 
  + 0.8660254037844387 d(0) + 0.8660254037844387 d(1) 
  + 0.8660254037844387 d(2) + 0.49999999999999994
\end{repl}

It is easy to see that terms of degree $i$ have the coefficients $\sin^{(i)}(\pi/6)$.
Since our implementation is polymorphic, if we apply the same function to the type for symbolic computation, say \hask{Symbolic}, we can reconstruct symbolic differentiation and check that the result is indeed correct symbolically:

\begin{repl}
>>> x :: Weil w Symbolic
>>> x = injectCoeff (var "x")

>>> normalise <$> sin (x + d0+d1+d2)
((-1.0) * cos x) d(0) d(1) d(2) + (- (sin x)) d(0) d(1) 
  + (- (sin x)) d(0) d(2) + (- (sin x)) d(1) d(2) 
  + (cos x) d(0) + (cos x) d(1) + (cos x) d(2) + (sin x)
\end{repl}

As stated before, the tensor-of-duals approach blows the number of terms exponentially.
Let us see how higher infinitesimal works.

\begin{code}
eps :: Floating a => Weil (DOrder 4) a
eps = di 0
\end{code}

Here, \hask{DOrder n} corresponds to an algebra $\R[X]/(X^n)$.
Note that, according to \Cref{lem:higher-infinitesimal}, to calculate an $n$\textsuperscript{th} derivative  we have to use $\R[X]/(X^{n+1})$.

\begin{repl}
>>> (sin (pi/6), cos (pi/6), -sin (pi/6)/2, -cos (pi/6)/6)
( 0.49999999999999994, 0.8660254037844387, 
 -0.24999999999999997, -0.14433756729740646)
  
>>> sin (pi/6 + eps) 
-0.14433756729740646 d(0)^3 - 0.24999999999999997 d(0)^2
  + 0.8660254037844387 d(0) + 0.49999999999999994

>>> normalise <$> sin (x + eps)
((-1.0) * cos x / 6.0) d(0)^3 + ((- (sin x)) / 2.0) d(0)^2
  + (cos x) d(0) + (sin x)
\end{repl}

Note that by \Cref{lem:higher-infinitesimal}, each coefficient is not directly a differential coefficient, but divided by $k!$, that is $f(x + \varepsilon)
= \sum_{k \leq 3} \frac{f^{(k)}(x)}{k !} \varepsilon^k$.

Let us see how tensor products of higher Weil algebras can be used to multivariate higher-order partial derivatives.
Suppose we want to calculate partial derivatives of $f(x, y) = e^{2x} \sin y$ up to $(2, 1)$\textsuperscript{th} order.
\begin{repl}
eps1, eps2 :: Weil (DOrder 3 |*| DOrder 2) a
(eps1, eps2) = (di 0, di 1)

>>> f (2 + eps1) (pi/6 + eps2)
94.5667698566742 d(0)^2 d(1) + 54.59815003314423 d(0)^2
  + 94.5667698566742 d(0) d(1) + 54.59815003314423 d(0)
  + 47.2833849283371 d(1) + 27.299075016572115

>>> normalise <$> f (x + eps1) (y + eps2)
(4.0 * exp (2.0 * x) / 2.0 * cos y) d(0)^2 d(1)
  + (4.0 * exp (2.0 * x) / 2.0 * sin y) d(0)^2
  + (2.0 * exp (2.0 * x) * cos y) d(0) d(1)
  + (2.0 * exp (2.0 * x) * sin y) d(0)
  + (exp (2.0 * x) * cos y) d(1) + (exp (2.0 * x) * sin y)
\end{repl}
One can see that the coefficient of $d(0)^i d(1)^j$ corresponds exactly to the value $D^{(i,j)}f(x,y)/i!j!$.
In this way, we can freely compose multiple Weil algebra to calculate various partial derivatives modularly.

\subsection{Computation in General Weil Algebra}
All examples so far were about the predefined, specific form of a Weil algebra.
Here, we demonstrate that we can determine whether the given ideal defines Weil algebras with \Cref{alg:weil-test}, and do some actual calculation in arbitrary Weil algebra.

First, we see that \textsc{WeilTest} rejects invalid ideals:

\begin{repl}
-- R[X,Y]/(X^3 - Y), not zero-dimensional
>>> isWeil (toIdeal [x ^ 3 - y :: Q[x,y]])
Nothing

-- R[X]/(X^2 - 1), which is zero-dimensional but not Weil
>>> isWeil (toIdeal [x ^ 2 - 1 :: Q[x]])
Nothing
\end{repl}

Next, we try to calculate in arbitrary chosen Weil algebra, $W = \R[x,y]/(x^2 - y^3, y^4)$, whose corresponding meaning in AD is unclear but is a Weil algebra as a matter of fact.

\begin{repl}
i :: Ideal (Rational[x,y])
i = toIdeal [x ^ 2 - y ^ 3, y ^ 4]

>>> isWeil i
Just WeilSettings 
  {weilBasis =[[0,0],[0,1], ..., [3,0]]
  , nonZeroVarMaxPowers = [3,3]
  , weilMonomDic = 
      [([0,2],[0,0,0,1,0,0,0,0]), ..., ([1,3],[0,0,0,0,0,0,0,1])]
  , table = [((0,0),1),((1,3),d(0)^2), ..., ((3,4),d(0)^3)]
  }
\end{repl}

Let us see what will happen evaluating $\sin(a + d_0 + d_1)$, where $d_0 = [x]_I, d_1 = [y]_I$?

\begin{repl}
>>> withWeil i (sin (pi/4 + di 0 + di 1))
-2.7755575615628914e-17 d(0)^3 - ... + 0.7071067811865476 d(0) 
  + 0.7071067811865476 d(1) + 0.7071067811865475

>>> withWeil i (normalise <$> sin (x + di 0 + di 1))
((-1.0) * (- (sin x)) / 6.0 + (-1.0) * cos x / 6.0) d(0)^3
  + ... + (cos x) d(0) + (cos x) d(1) + (sin x)
\end{repl}

Carefully analysing each output, one can see that the output coincides with what is given by \Cref{thm:series-is-smooth} and \Cref{lem:quot-ring-ideal}.

\section{Discussions and Conclusions}\label{sec:concl}
We have illustrated the connection between automatic differentiation (AD) and $C^\infty$-rings, especially Weil algebras.
Methods of AD can be viewed as techniques to calculate partial coefficients simultaneously by partially implementing the $C^\infty$-lifting operator for certain $C^\infty$-ring.
Especially, Forward-mode AD computes the $C^\infty$-structure of the dual number ring $\R[d] = \R[X]/X^2$ and Tower-mode computes that of the formal power series ring $\Rseries$.

The dual number ring $\R[d]$ is an archetypical example of Weil algebra, which formalises the real line with nilpotent infinitesimals.
We generalised this view to arbitrary Weil algebras beyond dual numbers, enabling us to compute higher-order derivatives efficiently and succinctly.
We gave general algorithms to compute the $C^\infty$-structure of Weil algebras.
With tensor products, one can easily compose (univariate) higher-order AD corresponding to Weil algebras into multivariate ones.

In this last section, we briefly discuss the possible applications other than AD, related works and future works.

\subsection{Possible Applications and Related Works}
Beside the reformulation of AD, we can argue that our methods can be used for a pedagogical purpose in teaching \emph{Smooth Infinitesimal Analysis} (SIA) and \emph{Synthetic Differential Geometry} (SDG).
In those fields, arguing in the appropriate intuitionistic topos, various infinitesimal spaces corresponding to Weil algebra is used to build a theory, expressed by the following \emph{generalised Kock-Lawvere axiom}~\cite{Moerdijk:1991aa}:

\begin{quote}
  For any Weil algebra $W$, the following evaluation map gives an isomorphism:
  \begin{align*}
    \mathop{\mathrm{ev}}: &W \to   \R^{\mathop{\mathrm{Spec}}_\R{W}}\\
       &a \mapsto \lambda f. f(a)
  \end{align*}
\end{quote}
This is another way to state the fact that Weil algebras are $C^\infty$-rings, viewed within some topoi.
For dual numbers, their meaning is clear: it just couples a value and their (first-order) differential coefficient.
However, solely from Kock-Lawvere axiom, it is unclear what the result is in general cases.
With the algorithms we have proposed, students can use computers to calculate the map given by the axiom.
In SIA and SDG, there are plenty of uses of generalised infinitesimal spaces such as $\R[x_1,\dots, x_n]/\braket{x_i x_j | i, j \leq n}$ or $\R[x]/(x^n)$.
Hence, concrete examples for these Weil algebras can help to understand the theory.

In the context of SDG, applying techniques in computer algebra to Weil algebras has attained only little interest.
One such example is Nishimura--Osoekawa~\cite{Nishimura:2007aa}: they apply zero-dimensional ideal algorithms to compute the generating relation of limits of Weil algebras.
However, their purpose is to use computer algebra to ease heavy calculations needed to develop the theory of SDG, and hence they are not interested in computing the $C^\infty$-structure of Weil algebras.

Implementing AD in a functional setting has a long history.
See, for example, Elliott~\cite{Elliott2009-beautiful-differentiation} for explanation and \texttt{ad} package by Kmett~\cite{Kmett:2010aa} for actual implementation.
In \texttt{ad} package, so-called \emph{Skolem trick}, or \emph{RankN trick} is applied to distinguish multiple directional derivatives.
We argue that our method pursues other direction of formulations; we treat higher infinitesimals as first-class citizens, enabling us to treat higher-order AD in a more modular and composable manner.

\subsection{Future Works}
In SDG, $C^\infty$-ring and higher infinitesimals are used as fundamental building blocks to formulate manifolds, vector fields, differential forms, and so on.
Hence, if one can extend our method to treat a general $C^\infty$-ring $C^\infty(M)$ of real-valued smooth functions on $M$, it can open up a new door to formulate differential geometry on a computer.
With such a formulation, we can define differential-geometric objects in more synthetic manner using nilpotent infinitesimals -- for example, one can define the tangent space $T_x M$ at $x \in M$ on some manifold $M$ to be the collection of $f: D \to M$ with $f(0) = x$, where $D$ is the set of nilpotents of order two.
Another virtue of such system is that we can treat infinitesimal spaces (derived from Weil algebras), manifolds, functions spaces, and vector spaces uniformly -- they are all living in the same category.
See Moerdijk--Reyes~\cite{Moerdijk:1991aa} for more theoretical details.
One major obstacle in this direction is that, even if $C^\infty(M)$ is finitely \emph{presented} as a $C^\infty$-ring, it is NOT finitely \emph{generated} as an $\R$-algebra, but $2^{\aleph_0}$-generated, by its very nature.
Hence, it seems impossible to compute $C^\infty(M)$ in purely symbolic and direct way; we need some workarounds or distinct formulations to overcome such obstacles.

As for connections with AD, there is also plenty of room for further exploration.
There are so many ``modes'' other than Forward- and Tower-modes in AD: for examples, Reverse, Mixed, etc.\ amongst others.
From the point of view of Weil algebras, they are just implementation details.
But such details matter much when one takes the efficiency seriously.
It might be desirable to extend our formulation to handle such differences in implementation method.
For such direction, Elliot~\cite{Elliott:2018aa} proposes the categorical formulation.
Exploring how that approach fits with our algebraic framework could be interesting future work, and perhaps also shed a light on the way to realise the aforementioned computational SDG.

\section*{Acknowledgments}
The author is grateful to Prof.\ Akira Terui, for encouraging to write this paper and many helpful feedbacks, and to anonymous reviewers for giving constructive comments.
\appendix
\section{Succinct Multivariate Lazy Tower AD}\label{sec:appendix}
For completeness, we include the referential implementation of the Tower-mode AD in Haskell, which can be used in \Cref{alg:smooth-weil}.
The method presented here is a mixture of Lazy Multivariate Tower~\cite{Pearlmutter:2007aa} and nested Sparse Tower~\cite{Kmett:2010aa}.
For details, we refer readers to the related paper by the author in RIMS K\^{o}ky\^{u}roku~\cite{Ishii:2021ab}.

The idea is simple: we represent each partial derivative as a path in a tree of finite width and infinite heights.
A path goes down if the function is differentiated by the 0\textsuperscript{th} variable.
It goes right if there will be no further differentiation w.r.t. 0\textsuperscript{th} variable, but differentiations w.r.t.\ remaining variable can take place.
This intuition is depicted by the following illustration of the ternary case:

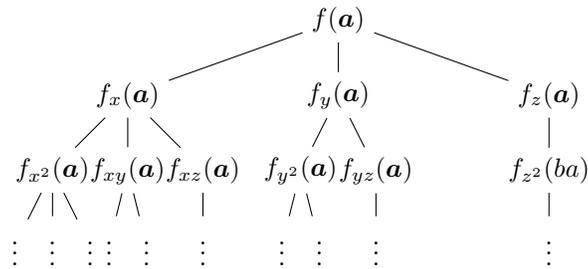
\begin{center}
\begin{tikzpicture}[
  level/.style={level distance=1cm},
  level 1/.style={sibling distance=2.8cm},
  level 2/.style={sibling distance=1cm},
  level 3/.style={sibling distance=5mm}
  ]
  \newcommand{\ba}{\boldsymbol{a}}
  \node (fa) {$f(\ba)$}
    child{ 
      node (fxa) {$f_x(\ba)$}
        child { 
          node (fxxa) {$f_{x^2}(\ba)$}
          child {node{$\vdots$}}
          child {node{$\vdots$}}
          child {node{$\vdots$}}
        }
        child {
          node (fxya) {$f_{xy}(\ba)$}
          child {node{$\vdots$}} child {node{$\vdots$}}
        }
        child { 
          node (fxza) {$f_{xz}(\ba)$}
          child { node{$\vdots$} }
        }
    }
    child {
      node (fya) {$f_y(\ba)$}
      child { node {$f_{y^2}(\ba)$} child {node{$\vdots$}} child{ node{$\vdots$}} }
      child { node {$f_{yz}(\ba)$} child {node {$\vdots$}} }
    }
    child {
      node (fza) {$f_z(\ba)$}
      child {
          node {$f_{z^2}(ba)$}
          child { node {$\vdots$} }
      }
    };
\end{tikzpicture}
\end{center}

This can be seen as a special kind of infinite trie (or prefix-tree) of alphabets $\partial_{x_i}$, with available letter eventually decreasing.

This can be implemented by a (co-)inductive type as follows:

\begin{code}
data STower n a where
  ZS :: !a -> STower 0 a
  SS :: !a -> STower (n + 1) a -> STower n a
     -> STower (n + 1) a
\end{code}

A tree can have an \emph{infinite height}.
Since Haskell is a lazy language, this won't eat up the memory and only necessary information will be gradually allocated.
Since making everything lazy can introduce a huge space leak, we force each coefficient \hask{a} when their corresponding data constructors are reduced to weak head normal form, as expressed by field strictness annotation \hask{!a}.

Then a lifting operation for univariate function is given by:

\begin{code}
liftSTower :: forall c n a.
  (KnownNat n, c a, forall x k. c x => c (STower k x) ) =>
  (forall x. c x => x -> x) ->
    -- ^ Function
  (forall x. c x => x -> x) ->
    -- ^ its first-order derivative
  STower n a ->
  STower n a
liftSTower f df (ZS a) = ZS (f a)
liftSTower f df x@(SS a da dus) 
  = SS (f a) (da * df x) (liftSTower @c f df dus)
\end{code}

Here, we use type-level constraint \hask{c} to represent to a subclass of smooth functions, e.g.\ $\hask{c} = \hask{Floating}$ for elementary functions.
Constraint of form $\forall x k.\ \texttt{c}\ x => \texttt{c}\ (\texttt{STower}\ k\ x)$ is an example of so-called \emph{Quantified Constraints}.
This requires \hask{c} to be implemented for any succinct Tower AD, provided that their coefficient type, say \hask{x}, is also an instance of \hask{c}.
This constraint is used recursively when one implements an actual implementation of instance \hask{c (STower n a)}.
For example, \hask{Floating} instance (for elementary floating point operations) can be written as follows:

\begin{code}
instance Floating a => Floating (STower n a) where
  sin = liftSTower @Floating sin cos
  cos = liftSTower @Floating cos (negate . sin)
  exp = liftSTower @Floating exp exp
  ...
\end{code}

In this way, we can implement Tower AD for a class of smooth function closed under differentiation, just by specifying an original and their first derivatives.

More general $n$-ary case of lifting operator is obtained in just the same way:

\begin{code}
liftNAry :: forall c n a m. 
    ( c a, forall x k. (KnownNat k, c x) => c (STower k x) ) =>
  -- | f, an m-ary smooth function
  (forall x. c x => Vec m x -> x) ->
  -- | partial derivatives of f,
  -- wrt. i-th variable in the i-th.
  Vec m (SmoothFun c m) ->
  Vec m (STower n a) ->
  STower n a
liftNAry f _ Nil = constSS $ f Nil
liftNAry f dfs xss =
  case sing @l of
    Zero -> ZS (f $ constTerm <$> xss)
    Succ (k :: SNat k) ->
      SS (f $ constTerm <$> xss)
         ( sum
         $ SV.zipWithSame 
            (\fi gi -> topDiffed gi * runSmooth fi xss)
           dfs xss
         )
         (liftNAry @c f dfs $ diffOther <$> xss)

diffOther :: STower (n + 1) a -> STower n a
diffOther (SS _ _ dus) = dus
\end{code}

\printbibliography

\end{document}